\documentclass[pra,superscriptaddress,twocolumn,nofootinbib]{revtex4}
\usepackage{amsmath,graphicx,epsfig,color,amsfonts}
\usepackage[hypertex]{hyperref}
\usepackage{amsthm}
\usepackage{url}

\def\be{\begin{equation}}
\def\ee{\end{equation}}
\def\bea{\begin{eqnarray}}
\def\eea{\end{eqnarray}}
\def\bma{\begin{mathletters}}
\def\ema{\end{mathletters}}

\def\0{\overline{0}}

\def\q0{\underline{0}}

\def\C{{\mathbb C}}
\def\id{{\mathbb I}}

\def\tr{\mbox{tr}}
\def\one{\leavevmode\hbox{\small1\normalsize\kern-.33em1}}

\def\bra#1{\langle#1|} \def\ket#1{|#1\rangle}
\def\braket#1#2{\langle#1|#2\rangle}

\def\proj#1{\ket{#1}\!\bra{#1}}

\newtheorem{theo}{Theorem}

\newtheorem{prop}[theo]{Proposition}

\def\id{{\mathbb I}}
\def\tr{\mbox{tr}}

\definecolor{nblue}{rgb}{0.2,0.2,0.9}
\definecolor{ngreen}{rgb}{0.2,0.55,0.2}
\definecolor{nred}{rgb}{0.9,0.2,0.2}
\definecolor{norange}{rgb}{1,0.5,0}
\definecolor{nblack}{rgb}{0,0,0}

\begin{document}

\title{Certifying entangled measurements in known Hilbert spaces}

\author{Tam\'as V\'ertesi}
\affiliation{Institute of Nuclear Research of the Hungarian
Academy of Sciences}

\author{Miguel Navascu\'es}
\affiliation{Facultad de Matem\'aticas, Universidad Complutense de Madrid}

\begin{abstract}
We study under which conditions it is possible to assert that a
joint demolition measurement cannot be simulated by Local
Operations and Classical Communication. More concretely, we
consider a scenario where two parties, Alice and Bob, send each an
unknown state to a third party, Charlie, who in turn interacts
with the states in some undisclosed way and then announces an
outcome. We show that, under the assumption that Alice and Bob
know the dimensionality of their systems, there exist situations
where the statistics of the outcomes reveals the nature of
Charlie's measurement.

\end{abstract}

\maketitle

\section{Introduction}

Quantum nonlocality, the fact that entangled quantum systems
separated in space can violate Bell inequalities \cite{Bell64}, is
one of the most profound discoveries in science. This nonlocal
nature of entanglement has also been identified as an essential
resource for various quantum information processing tasks. It
allows to reduce communication complexity \cite{review} or makes
it possible to devise device-independent protocols for quantum key
distribution \cite{key}, genuine random number generations
\cite{rng}, and state tomography \cite{tomo}. Indeed, these tasks
can be performed without resorting to the actual inner working of
the devices, which is a feat without counterpart in the classical
world.

In all of the above instances, quantum nonlocality arises when
space-like separated measurements are performed over jointly
prepared entangled states. However, the existence of nonlocality
can also be manifested in a dual setup
\cite{PeresWootters},\cite{Bennett}: quantum states which are
prepared separately (hence unentangled) exhibit anomalous
correlations when measured jointly. In this case, more information
can be revealed by joint measurements than can be gained by using
any sequence of Local Operations assisted by Classical
Communication (LOCC). In particular, the superiority of globally
entangled measurements over LOCC ones has been proved conclusively
\cite{GisinPopescu}. The improved performance is due to
entanglement occurring in the process of measurement and again it
is not associated with the states.

For quantifying the effectiveness of entangled measurements,
different figures of merit have been proposed, like the Shannon
mutual information \cite{PeresWootters} or the quantum fidelity
\cite{GisinPopescu}. A common feature of those is, though, that
one must rely both on the precise \emph{form of the states}
prepared and on the \emph{a priori probabilities} of the
preparations; otherwise, the estimations on the figure of merit
would be not reliable. Hence, in order to ascertain that a
measurement is truly globally entangled and not only an LOCC one,
one must resort to a detailed knowledge of the preparation
procedure. Clearly, this issue sheds doubt on the reliability of
the obtained results.

In the present paper we take another approach, more in the spirit
of a black box scenario, and tackle the problem of certifying
entangled measurements by introducing a kind of dual Bell
inequalities. These inequalities involve correlation terms which
can be gathered from experimental data, and do not depend on the
specific form of the quantum states to be prepared. Actually, the
only necessary condition entering in the derivation is the Hilbert
space dimension of the prepared particles. Taking into account
that entangled measurements enable increased classical capacity of
quantum channels \cite{BS} and efficient quantum state estimation
\cite{MassarPopescu}, in the future it would also be interesting
to find applications related to these tasks within a black box
approach.

The notation we use along the article is introduced in the next
section. Section~\ref{nonlin} provides an instance of certifiable
entangled measurements based on a nonlinear inequality.
Section~\ref{lin} is devoted to a numerical study where linear
inequalities are applied to detect entangled measurements in the
simplest possible scenarios. Section~\ref{conc} summarizes our
conclusions.

\section{Preliminaries and notation}

Let us imagine the following communication scenario, close to the
spirit of the simultaneous message passing model \cite{SMP}. Two
separated parties, Alice and Bob, have each some preparation
device with $N$ possible inputs. Alice (Bob) is thus able to
prepare any of the unknown qudit states
$\{\rho_x\in\C^D\}_{x=1}^N$ ($\{\sigma_y\in\C^D\}_{y=1}^N$). Alice
and Bob then send their states to a third party, Charlie, who
performs a measurement over those two states, announcing a dit $a$
of dimension $K$ at the end of the process. Denote by
$\{M_a:a=0,...,K-1\}$ the elements of Charlie's Positive Operator
Valued Measure (POVM). Then, for each pair of inputs $x$ and $y$,
several repetitions of this primitive would allow Alice and Bob to
estimate the frequency of occurrence

\be
P(a|x,y)\equiv \tr(\rho_x\otimes\sigma_y\cdot M_a)
\ee

\noindent of each outcome $a=1,2,\ldots,K$. Depending on the form
of the POVM elements $\{M_a:a=0,...,K-1\}$, we will distinguish
four different types of measurements:

\vspace{10pt}
\emph{Classical measurements.} \noindent Suppose that Charlie's
measurement is fixed in advance to a given complete projective
measurement on each system (say, in the computational basis),
followed by some random processing of the raw output. Any set of
distributions $P(a|x,y)$ so generated can be reproduced in a
classical setting where Alice and Bob send each one dit of
information to Charlie, who, in turn, applies a random function
$f:\{0,...,D-1\}\times \{0,...,D-1\}\to K$ and then announces the
result. For each $NDK$ scenario, we will denote by $\tilde{P}^C$
the set of all possible correlations $P(a|x,y)$ that can be
attained in such a way.

It is easy to see that any point of $\tilde{P}^C$ is a convex
combination of deterministic classical strategies. As we will see,
though, the converse is not true. On the other hand, if Alice, Bob
and Charlie share random variables, then any convex sum of
deterministic strategies can be attained. Since for any $NDK$
scenario the number of deterministic classical strategies is
finite, the set of correlations attainable through classical maps
and shared randomness is a polytope, that we will denote by $P^C$.
This classical set is symbolized in Figure~\ref{sets} by a square.

\begin{figure}
\includegraphics[width=7.5 cm]{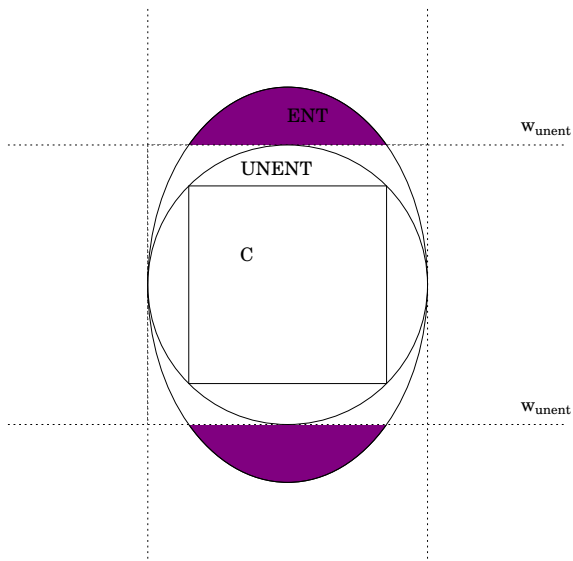} \caption{
A schematic picture of the 2-dimensional slice of the probability
space for different sets in the general $NDK$ scenario.
The square represents the classical set $P^C$, the circle and the
ellipse designate the set of unentangled ($P^{unent}$) and the set
of entangled measurements ($P^{ent}$), respectively. In the purple region
entangled measurements are witnessed by linear inequalities.} \label{sets}
\end{figure}

\vspace{10pt}

\emph{Unentangled and LOCC measurements.} \noindent If each of the
POVM elements $M_a$ is a separable operator, we will say that the
measurement $M$ is unentangled. We will say that Charlie's
measurement can be attained via Local Operations and Classical
Communications (LOCC) if $M$ corresponds to a sequence of local
measurements on Alice's and Bob's individual qubits, with each
measurement possibly depending on the outcomes of earlier
measurements. It is known that the class of unentangled
measurements is strictly greater than the class of LOCC
measurements \cite{Bennett}. This implies that
$\tilde{P}^{unent}\supseteq \tilde{P}^{LOCC}$, where
$\tilde{P}^{unent},\tilde{P}^{LOCC}$ resp. denote the sets of
distributions $P(a|x,y)$ attainable through unentangled and LOCC
measurements in a given $NDK$ scenario. It is an open question
whether such an inclusion is strict.

As before, if we allow Alice, Bob and Charlie share some prior
random variables, the resulting sets of correlations
$P^{unent}\supseteq P^{LOCC}$ are convexifications of the former.
A two-dimensional slice of the unentangled set $P^{unent}$ is
symbolized in Figure~\ref{sets} by a circle.

\vspace{10pt}

\emph{General measurements.} \noindent Here the measurement
operators are only limited by positivity and normalization
constraints, i.e.,
\be
M_a\succeq 0,\mbox{ }\sum_{a=0}^{K-1}M_a=\id,
\ee

\noindent where the operators $M_a\in B(\C^D\otimes \C^D)$ may
very well be entangled. These are the most general measurements
allowed by quantum mechanics. In analogy with the former sets, we
will denote by $\tilde{P}^{ent}$ and $P^{ent}$ the set of
distributions accessible through general joint measurements and
its convex hull. The latter is depicted as an ellipse in
Figure~\ref{sets}.

\vspace{10pt}

The numerical study of the sets $\tilde{P}^C$, $\tilde{P}^{unent}$
and $\tilde{P}^{ent}$ is much more convoluted than that of their
convexifications $P^C$, $P^{unent}$ and $P^{ent}$. Consequently,
most of this article is devoted to the numerical analysis of the
latter.

If we reflex a bit about the previous definitions, it will soon be
clear that the above sets of correlations can only be
distinguished experimentally in a black box way if the $NDK$
scenario is such that $N>D$. Otherwise, by forcing Alice and Bob
to send orthogonal states, Charlie could always infer the values
of $x,y$ through an appropriate projective measurement, and so
\emph{any} conceivable distribution $P(a|x,y)$ could be
classically realized. The simplest scenario where we may expect to
find an interesting structure is thus the 322 scenario. And,
indeed, the next example shows that there the sets $\tilde{P}^C$
and $\tilde{P}^{unent}$ differ from $\tilde{P}^{ent}$. It also
proves that neither of the former two sets is convex, and so the
convexifications $P^C,P^{LOCC},P^{unent}$ are non-trivial.

\section{An instance of entanglement detection}\label{nonlin}

Consider the 322 scenario, and denote by $P$ the matrix of
probabilities $P_{xy}\equiv P(0|x,y)$. Suppose now that
\begin{equation}
P=\begin{pmatrix}
 0 &  1/4 & 1/4\\
 1/4 & 0 & 1/4\\
 1/4 & 1/4 & 0
\end{pmatrix}.
\label{nonlinear}
\end{equation}

\noindent This set of probabilities can be attained if the states
$\{\rho_x,\sigma_y\}$ and the POVM element $M\equiv M_{a=0}$
correspond to
\begin{align}
\rho_1&=\proj{0},&\rho_2&=\proj{+},&\rho_3&=\proj{+i},\nonumber\\
\sigma_1&=\proj{1},&\sigma_2&=\proj{-},&\sigma_3&=\proj{+i},
\end{align}
\begin{equation}
M=\proj{\Psi^+},
\end{equation}

\noindent where
$\ket{\Psi^+}=\frac{1}{\sqrt{2}}(\ket{00}+\ket{11})$. It hence
follows that $P\in \tilde{P}^{ent}$.

We will prove that $M$ must be an entangling measure (i.e.,
$P\not\in \tilde{P}^{unent}$) by \emph{reductio ad absurdum}.
Imagine, therefore, that $M$ is separable, i.e.,
$M=\sum_{i=1}^K\lambda_i\proj{u_i}\otimes\proj{v_i}$. Then, the
condition $P_{11}=0$ implies that: 1) $\rho_1$ or $\sigma_1$ (or
both) is a pure quantum state; and 2) there exists a subset of
indices ${\cal I}\subset \{1,2,...,K\}$ such that
\be
M=\sum_{i\in {\cal I}}\lambda_i\rho_1^\perp\otimes\proj{v_i}+\sum_{i\not\in {\cal I}}\lambda_i\proj{u_i}\otimes \sigma_1^\perp.
\label{decomp}
\ee

\noindent Here $\omega^\perp$ denotes $\id_2-\omega^\perp$, for
any qubit state $\omega$. From
$\tr(M\rho_1\otimes\sigma_1)=P_{11}\not=P_{12}=\tr(M\rho_1\otimes\sigma_2)$
we have that $\sigma_1\not=\sigma_2$, and thus
$\tr(\sigma_1\sigma_2^\perp)\not=0$. Likewise,
$\tr(\rho_x\rho^\perp_{x'}),\tr(\sigma_y\sigma^\perp_{y'})\not=0$,
for all $y\not=y',x\not=x'$. This implies that
\begin{eqnarray}
&&\tr(M\rho_2\otimes\sigma_2)=\sum_{i\in {\cal I}}\lambda_i\tr(\rho_1^\perp\rho_2)\tr(\proj{v_i}\sigma_2)\nonumber\\&&+\sum_{i\not\in {\cal I}}\lambda_i\tr(\proj{u_i}\rho_2)\tr(\sigma_1^\perp\sigma_2)\nonumber\\
&&\geq\delta\sum_{i\in {\cal I}}\lambda_i\tr(\proj{v_i}\sigma_2)+\delta\sum_{i\not\in {\cal I}}\lambda_i\tr(\proj{u_i}\rho_2),
\end{eqnarray}

\noindent for some $\delta>0$. The condition $P_{22}=0$ therefore
requires that $\tr(\proj{v_i}\sigma_2)=0,i\in{\cal I}$, and
$\tr(\proj{u_i}\rho_2)=0,i\not\in {\cal I}$, so
\be
M=\sum_{i\in {\cal I}}\lambda_i\rho_1^\perp\otimes\sigma_2^\perp+\sum_{i\not\in {\cal I}}\lambda_i\rho_2^\perp\otimes \sigma_1^\perp.
\label{decomp2}
\ee

\noindent Using this last decomposition, we get
\begin{align}
\tr(M\rho_3\otimes\sigma_3)=&\sum_{i\in {\cal I}}\lambda_i\tr(\rho_1^\perp\rho_3)\tr(\sigma_2^\perp\sigma_3)\nonumber\\&+\sum_{i\not\in {\cal I}}\lambda_i\tr(\rho_2^\perp\rho_3)\tr(\sigma_1^\perp\sigma_3)\nonumber\\
&\geq \delta'\sum_{i=1}^K\lambda_i\geq\delta'\tr(M\rho_1\otimes\sigma_2)>0.
\end{align}

\noindent with $\delta'>0$. But $P_{33}=0$, from which it follows
that $M$ is indeed entangling. Actually, it is possible to obtain a
non-linear witness to detect $M$'s entanglement if we just follow
the previous steps carefully while keeping track of the
approximation errors (see the Appendix).

We have thus proven that $P\not\in \tilde{P}^{unent}$, i.e., the
matrix of probabilities $P$ cannot be observed if only unentangled measurements are allowed. However, the matrix
$P_{xy}$ can be expressed as a convex combination of four
deterministic classical strategies. Indeed,
\begin{equation}
P=\frac{1}{4}\begin{pmatrix}
 0 & 1 & 1\\
 0 & 0 & 0\\
 0 & 0 & 0
\end{pmatrix}+
\frac{1}{4}\begin{pmatrix}
 0 &   0 & 0\\
 1 &   0 & 1\\
 0 &   0 & 0
\end{pmatrix}+
\frac{1}{4}\begin{pmatrix}
 0 &  0 & 0\\
 0 &  0 & 0\\
 1 &  1 & 0
\end{pmatrix}+
\frac{1}{4}\begin{pmatrix}
 0 &   0 & 0\\
 0 & 0 & 0\\
 0 & 0 & 0
\end{pmatrix}.
\label{convexcomb}
\end{equation}

The sets $\tilde{P}^{C},\tilde{P}^{LOCC},\tilde{P}^{unent}$ are
hence not convex, since otherwise $P$ would belong to them.
Allowing the three parties to have shared randomness thus leads to
a radically different scenario.

\section{Numerical study of $P^C$, $P^{unent}$ and $P^{ent}$ in the $N22$
scenario}\label{lin}

\subsection{Procedure}

Since convex sets can be completely characterized by systems of
linear inequalities, it is legitimate to investigate how different
linear constraints on $P^C$ of the form
$\sum_{x,y}W_{xy}P_{xy}\leq w_c$ may be violated by elements of
$P^{unent}$ and $P^{ent}$.

This leads us to the problem of maximizing quantities like
\begin{equation}
\sum_{x,y=1}^{N}{W_{x,y}P_{x,y}}=\sum_{x,y=1}^N{W_{xy}\tr(\rho_x\otimes\sigma_y
M)}
\label{maxi}
\end{equation}
over all possible POVM elements $M$ in the unentangled or the
general class, and over all qubit states $\rho_x$ and $\sigma_y$
so as to get the values $w_{unent}$ and $w_{ent}$. Clearly, if
$w_{ent}>w_{unent}$ for some matrix $W$, then there exist
experimental situations where one can prove that the statistics
observed in the lab cannot be simulated with unentangled measurements and shared randomness.

Numerical optimization to obtain $w_{unent}$ is carried out very
similarly to the iterative algorithm used in \cite{activation}.
The steps are the following:
\begin{enumerate}
\item Generate some random pure qubit states $\rho_x$, $\sigma_y$,
$x,y=1,2,\ldots,N$. \item In dimension $D=2$, for fixed states
$\rho_x,\sigma_y$, maximizing Eq.~(\ref{maxi}) reduces to a
semidefinite problem: define thus
$F\equiv\sum_{x,y=1}^NW_{x,y}{\rho_x\otimes\sigma_y}$, and
maximize $\tr(MF)$, subject to $M\succeq 0,\one-M\succeq
0,PT(M)\succeq 0,\one-PT(M)\succeq 0$, where $PT$ denotes Partial
Transposition \cite{peres,horodecki}. \item Fix $M$ and $\sigma_y$
and maximize $\tr(G_x\rho_x)=\langle\psi_x|G_x|\psi_x\rangle$ for
$x=1,2,\ldots,N$ where $G_x=\tr_B(\sum_y \one\otimes\sigma_y M)$.
This amounts to find the maximum eigenvalue of $G_x$ with the
corresponding eigenvector $|\psi_x\rangle$ for each $x$. \item For
fixed $M$ and $\rho_x$, the optimal $\sigma_y$ can be found such
as $\rho_x$ in step 3. \item If convergence has not been achieved,
go back to step 2.
\end{enumerate}

This optimization algorithm may encounter several local optima,
therefore we must iterate it many times in order to ascertain with
reasonable confidence that the largest maximum has been found. In
Sections \ref{sec322} and \ref{sec422} the semidefinite programs
described in step 2 were solved using the SeDuMi package
\cite{sedumi}.

A very similar procedure can be applied if we wish to inquire
about the maximal value $w_{ent}$ which can be achieved within the
class of general quantum measurements. The only difference is in
step 2, where we have to maximize $\tr(MF)$ subject to the only
constraint $0\preceq M \preceq\one$. However, in this case the
maximum for a fixed $F$ can be easily found by diagonalization
$F=\sum_{i=1}^4\lambda_i|\phi_i\rangle\langle\phi_i|$, resulting
in the optimal
$M=\sum_{i=1}^4{\frac{\mbox{sgn}(\lambda_i)+1}{2}}|\phi_i\rangle\langle\phi_i|$.
Thereby the optimization algorithm for the case of entangled
measurements is much faster.

\subsection{The 322 scenario}
\label{sec322} We next describe the way linear witnesses were
produced for entangled measurements. Our starting point is the
classical polytope $P^C$ for the $322$ scenario, consisting of 104
extremal points in dimension 9. All its facets can be enumerated
by using the double description method implemented in Fukuda's cdd
package \cite{cdd}. The polytope $P^C$ has 1230 facets; among
them, 9 describe positivity facets $P_{x,y}\ge 0$, $x,y=1,2,3$.
The rest of them define non-trivial inequalities $W$. However,
most of them are equivalent under the following operations,
\begin{itemize}
\item permuting Alice's inputs $x$ \item permuting Bob's inputs
$y$ \item exchanging parties Alice and Bob \item multiplying by
$-1$ all coefficients $W_{x,y}$ of the matrix $W$.
\end{itemize}

It turns out that there are 13 inequivalent facets of the polytope
$P^C$. Those are listed in Table~\ref{table1}, where we also give
the maximum classical value $w_c$, the value $w_{unent}$
achievable with unentangled measurements, the value $w_{ent}$
attainable with general measurements, and the number of
vertices lying on the particular facet. We must stress that the
values for $w_{unent}$ and $w_{ent}$ come from numerical
optimization, and so only local optimality is guaranteed.
Nevertheless, due to the small dimensionality of the problem and
the number of iterations of the main algorithm, we are quite
confident about their overall optimality as well. Interestingly,
we found that in each of the 13 cases $w_{unent}$ could be
achieved with rank-2 projective measurement operators.

\begin{table*}[t]
\caption{Results for the $13$ facet inducing inequalities in the
$322$-scenario.}\vskip 0.2truecm \centering
\begin{tabular}{l l l l c | c c c c c c c c c}
\hline\hline Case & $w_c$ & $w_{unent}$ & $w_{ent}$ & vertices &
$W_{11}$ & $W_{12}$ & $W_{13}$ & $W_{21}$ & $W_{22}$ &
$W_{23}$ & $W_{31}$ & $W_{32}$ & $W_{33}$\\
\hline
1 &    6 & 6.4006 &6.4006  & 11    &    -2  &  -1 &    1 &    2 &    4 &   -2 &    4  &  -5  &  -1\\
2 &    2 & 2.8284 &2.8284  & 13    &    -2  &   0 &    2 &   -2 &    1 &   -1 &    0  &   1  &  -1\\
3 &    1 & 1.4142 &1.4142  & 20    &    -1  &  -1 &    0 &    0 &    1 &    0 &    1  &  -1  &   0\\
4 &    2 & 2.3371 &2.5     & 16    &    -1  &  -1 &    0 &   -1 &    1 &    1 &    1  &  -1  &   1\\
5 &    1 & 1.3371 &1.3371  & 12    &    -1  &  -1 &    1 &   -1 &    0 &   -1 &    0  &   1  &   0\\
6 &    2 & 2.6742 &2.6742  & 10    &    -2  &  -1 &    2 &   -2 &    1 &   -1 &    0  &   1  &   0\\
7 &    2 & 2.1623 &2.1623  &  9    &    -1  &   0 &    1 &    0 &    2 &   -2 &    1  &  -2  &  -1\\
8 &    2 & 2.1403 &2.1403  & 10    &    -2  &  -1 &    1 &    0 &    2 &   -2 &    2  &  -3  &  -1\\
9 &    1 & 1.2058 &1.2058  & 10    &    -1  &  -1 &    1 &    0 &    1 &   -1 &    1  &  -2  &  -1\\
10&    2 & 2.7275 &2.7275  & 14    &    -3  &  -2 &    1 &   -2 &    2 &    0 &    1  &  -1  &   1\\
11&    1 & 1.3094 &1.3094  & 16    &    -2  &  -2 &    1 &   -1 &    1 &    0 &    1  &  -1  &   0\\
12&    3 & 3.8467 &3.8467  & 11    &    -5  &  -3 &    2 &   -3 &    3 &   -1 &    2  &  -1  &   1\\
13&    2 & 2.1186 &2.1186  &  9    &    -3  &  -1 &    2 &   -1 &    2 &   -3 &    2  &  -3  &  -1\\
 \hline
 \end{tabular}
 \label{table1}
 \end{table*}

We now focus on inequality $\#4$ from Table~\ref{table1}, which is
the only inequality for which $w_{ent}>w_{unent}$, and hence it
enables to witness entangled measurements. The inequality looks (in an equivalent form) as
\begin{align}
&-P_{11}-P_{12}+P_{13}+P_{21}+P_{23}+P_{31}-P_{32}-P_{33}\nonumber\\
&\le\frac{2+3\sqrt 6}{4}\simeq2.3371, \label{wit322}
\end{align}
and can be violated by entangled measurements up to the value of
$2.5$. Next, we present the actual measurements and states
achieving $w_{ent}$ and $w_{unent}$.

\emph{Entangled case.} \noindent The POVM element $M$ is a rank-2
projector
\begin{equation}
M=|m\rangle\langle m|+|m^\perp\rangle\langle m^\perp|,
\end{equation}
with
\begin{align}
|m\rangle&=(0,\cos\theta,\sin\theta,0)\nonumber\\
|m^\perp\rangle&=\left(\frac{\cos\phi}{\sqrt
2},-\sin\phi\sin\theta,\sin\phi\cos\theta,-\frac{\cos\phi}{\sqrt
2}\right),
\end{align}
having $\theta=2\arctan(\frac{\sqrt{10} -1}{3})$ and $\phi=\pi/3$.
The partial transpose of $M$ has the following eigenvalues:
\begin{equation}
\lambda=\left(\frac{5-\sqrt{41}}{10},\frac{1}{5},\frac{4}{5},\frac{5+\sqrt{41}}{10}\right),
\end{equation}
with the first entry being negative. Both Alice's and
Bob's states are pure and real valued,
$\rho_x=|\psi_x\rangle\langle\psi_x|$ and
$\sigma_y=|\phi_y\rangle\langle\phi_y|$, respectively. Their
explicit forms are, respectively,
\begin{align}
|\psi_1\rangle&=(1,0),\nonumber\\
|\psi_2\rangle&=(\cos\alpha,\sin\alpha),\nonumber\\
|\psi_3\rangle&=(\cos 2\alpha,\sin 2\alpha)\nonumber,
\end{align}
and $|\phi_y\rangle=|\psi_y\rangle$ for $y=1,2,3$ with
$\alpha=2\arctan\left(\frac{2\sqrt 10-\sqrt 15}{5}\right)\simeq
0.9117$ rad. With these values in hand, the probability matrix
$P$ with components $P_{x,y}=\tr(\rho_x\otimes\rho_y M)$
looks as follows:
\begin{equation}
 P=\begin{pmatrix}
 0.125 &  0.25 & 0.875\\
 0.25 & 0.5 & 0.75\\
 0.875 & 0.25 & 0.125
\end{pmatrix}.
\label{Pent}
\end{equation}
And, indeed, evaluating the left-hand side of
inequality~(\ref{wit322}), we get $w_{ent}$=2.5.

\emph{Unentangled case.} \noindent In appropriate local bases, a
separable rank-2 projective operator can be written in the form
\begin{equation}
M=|0\rangle\langle 0|\otimes |m\rangle\langle m|+|1\rangle\langle
1|\otimes |0\rangle\langle 0|,
\end{equation}
\noindent where $|m\rangle=(\cos\theta,\sin\theta)$, and in
our particular case, $\theta=\frac{1}{2}\arccos\frac{-1}{4}\simeq 0.911738$ rad. The
states $\rho_x=|\psi_x\rangle\langle\psi_x|$ and
$\sigma_y=|\phi_y\rangle\langle\phi_y|$, on the other hand, are
given by
\begin{align}
|\psi_1\rangle&=(1,0),& |\phi_1\rangle&=(1,0),\nonumber\\
|\psi_2\rangle&=(1,0),& |\phi_2\rangle&=(\cos\alpha,\sin\alpha),\nonumber\\
|\psi_3\rangle&=(0,1),& |\phi_3\rangle&=(\cos-\alpha,\sin-\alpha),
\end{align}
with
$\alpha=-\arctan\left(\sqrt{\frac{8}{5}}+\sqrt{\frac{3}{5}}\right)\simeq
-1.11493$ rad. With these, the probability vector is
\begin{equation}
 P=\begin{pmatrix}
 0.375 &  P_{12} & P_{13}\\
 0.375 &  P_{12} & P_{13}\\
 1 & P_{12} & P_{12}
\end{pmatrix},
\label{Punent}
\end{equation}
where $P_{12}=\frac{5}{4(4+\sqrt 6)}\simeq 0.193814$ and
$P_{13}=\frac{P_{12}}{2}(5+2\sqrt 6)\simeq 0.959279$. Plugging
these numbers into the left-hand side of
inequality~(\ref{wit322}), we get $w_{unent}=\frac{2+3\sqrt
6}{4}\simeq 2.3371$. In some sense, the probability
vector~(\ref{Punent}) approximates the best way among
unentangled measurements the probability vector~(\ref{Pent})
corresponding to an entangled measurement.

\subsection{The 422 scenario}
\label{sec422} The classical polytope $P^C$ for the 422 scenario
is spanned by 520 non-redundant vertices in the 16 dimensional
probability space. However, it turned out that the full
characterization of this polytope in terms of facets is
computationally an elusive task. Instead, the problem was
approached in a different way. We scanned through all the
inequalities with small integer coefficients ($-1,0,+1$), and
sorted out all of them which are facet defining. Several
inequivalent facets have been found in this way (beyond the ones,
which are just straightforward extensions of the $322$-type
inequalities in Table~\ref{table1}). In Table~\ref{table2}, we
list all those 10 inequivalent facets for which
$w_{ent}>w_{unent}$, hence witnessing entangled measurements.
Among them, $\#9$ is equivalent to the single $322$-type witness
in (\ref{wit322}). Table~\ref{table2} contains results on $w_c$,
$w_{unent}$, $w_{ent}$, and also gives the sole negative
eigenvalue $\lambda_1$ of $PT(M)$ achieving $w_{ent}$. Since all
the other eigenvalues are positive, the absolute value
$|\lambda_1|$ defines negativity, a valid entanglement measure
\cite{GT}.

Inequality $\#2$ is interesting on its own. In this case,
$w_{ent}=3.3195$ is attained with a negativity of $0.1744$,
although a larger value of negativity of $0.1777$ is found at
another locally optimal point with $w_{ent}'=3.145$. This feature
resembles the dual case obtained in a Bell scenario, where for
several tight 2-party Bell inequalities the maximum quantum value
was achieved with non-maximally entangled 2-qubit states
\cite{partial}. It has also been observed that for all the $422$
witnesses for which unentangled measurements gave numerically maximal violation, the measurements could always be brought
to a form of rank-2 projective matrices.

\begin{table*}[t]
\caption{Results for the $10$ facet inducing inequalities
witnessing entangled measurements in the $422$-scenario.}\vskip
0.2truecm \centering
\begin{tabular}{l l l l c c | c c c c c c c c c c c c c c c c}
\hline\hline Case & $w_c$ & $w_{unent}$ & $w_{ent}$ & vertices &
$\lambda_1$ & $W_{11}$ & $W_{12}$ & $W_{13}$ & $W_{14}$ & $W_{21}$
& $W_{22}$ & $W_{23}$ & $W_{24}$ & $W_{31}$ & $W_{32}$ & $W_{33}$
& $W_{34}$ & $W_{41}$ & $W_{42}$ & $W_{43}$ & $W_{44}$\\
\hline
1 & 2 & 2.3371 & 2.3510 & 31 & -0.1238 & -1&-1&-1&0&-1&0&1&0&-1&1&-1&1&1&-1&0&1\\
2 & 3 & 3.2361 & 3.3195 & 34 & -0.1744 & -1  &  -1  &  -1  &   0  &  -1  &   0  &   1  &   1  &   0  &   1  &  -1  &   1  &   1  &  -1  &   0  &   1\\
3 & 2 & 2.3371 & 2.4369 & 38 & -0.1154 & -1  &  -1  &  -1  &   1  &  -1  &   0  &   1  &  -1  &   0  &   0  &   0  &   0  &   0  &   1  &   0  &   1\\
4 & 2 & 2.3371 & 2.4339 & 25 & -0.1133 & -1  &  -1  &  -1  &   1  &  -1  &   0  &   1  &  -1  &   0  &   0  &   1  &   1  &   1  &  -1  &   0  &  -1\\
5 & 1 & 1.3371 & 1.3413 & 28 & -0.1016 & -1  &  -1  &   0  &   0  &  -1  &   0  &  -1  &   0  &   0  &   0  &   0  &   1  &   0  &   1  &  -1  &  -1\\
6 & 2 & 2.3371 & 2.4322 & 24 & -0.1089 & -1  &  -1  &   0  &   0  &  -1  &   0  &  -1  &   1  &   0  &   1  &  -1  &  -1  &   1  &   0  &   0  &   1\\
7 & 2 & 2.1623 & 2.3028 & 19 & -0.0718 & -1  &  -1  &   0  &   0  &  -1  &   1  &  -1  &   1  &   1  &   0  &  -1  &   0  &   1  &   0  &   0  &  -1\\
8 & 2 & 2.1623 & 2.3028 & 26 & -0.0718 & -1  &  -1  &   0  &   0  &  -1  &   1  &   0  &   0  &   1  &  -1  &  -1  &   1  &   1  &   1  &  -1  &  -1\\
9 & 2 & 2.3371 & 2.5    & 72 & -0.1403 & -1  &  -1  &   0  &   0  &  -1  &   1  &   0  &   1  &   0  &   0  &   0  &   0  &   1  &  -1  &   0  &   1\\
10& 2 & 2.2058 & 2.3773 & 22 & -0.0982 & -1  &  -1  &   0  &   0  &  -1  &   1  &   0  &   1  &   0  &   1  &  -1  &  -1  &   1  &  -1  &  -1  &   1\\
 \hline
 \end{tabular}
 \label{table2}
 \end{table*}

To conclude this section, the $422$ scenario gave us a few new
entangled measurement witnesses over the single witness of
(\ref{wit322}) of the $322$ scenario. However, no significant
improvement could be found in the efficiency of detecting
entangled measurements with respect to unentangled measurements.
Also, no simple generalization of inequality~(\ref{wit322}) for
larger input alphabets was discovered among the $422$
inequalities.

\section{Conclusion}\label{conc}

In this paper, we have studied how to translate the concept of
witnessing entangled measurements to the device-independent arena.
We found that even in a black-box scenario it is possible for a
theorist to assert that a joint demolition measurement cannot be
simulated with any sequence of local measurements assisted by
classical communication: the only assumption one has to rely on is
the dimensionality of the probe states. Following this line of
thought, we derived correlation inequalities which allow to test
the superiority of entangled measurements over unentangled ones in
two-qubit experiments where 9 (Secs.~\ref{nonlin},\ref{sec322})
and 16 correlation terms (Sec.~\ref{sec422}) are estimated.

Previous routes to certify the `non-locality' of unknown measuring
devices involved carrying out complete tomography of their
associated POVM elements \cite{MM}, or envisaging a state
discrimination problem where entangled measurements have an
advantage with respect to unentangled ones
\cite{PeresWootters,GisinPopescu}. Either scheme has been
conducted experimentally \cite{MM,exp} with success. Both
approaches, though, relied crucially on a detailed knowledge of
the probe states to be measured, and so their conclusions cannot
be considered definite.

In contrast, in this article we have proven that the experimental
inaccuracies arising from an imperfect preparation of the probe
states can be completely eliminated by certifying entangled
measurements within a black box approach. Given the simplicity of
our inequalities, we thus propose their experimental
implementation as an interesting challenge.

\begin{acknowledgements}
We thank David P\'erez-Garc\'ia for valuable discussions. This
work has been supported by a J\'anos Bolyai Programme of the
Hungarian Academy of Sciences and by the European project
QUEVADIS.
\end{acknowledgements}

\appendix

\section{A Non-linear witness}

In order to obtain a witness to certify the non-locality of $P$ in
Eq.~(\ref{nonlinear}), we will make use of the following result.

\begin{prop}
\label{prob} Let $A\in B(\C^d)$ satisfy $0\preceq A\preceq \id_d$,
and let $\omega_1,\omega_2\in B(\C^d)$ be two normalized quantum
states, with $\tr(A\omega_i)=P_i$, for $i=1,2$. Then,
\be
\tr(\omega_1\cdot\omega_2)\leq \left(f(P_1,P_2)\right)^2,
\label{formu} \ee
\end{prop}

\noindent where
\be f(P_1,P_2)\equiv\sqrt{P_1\cdot P_2}+
\sqrt{(1-P_1)\cdot (1-P_2)}. \label{deff} \ee

\noindent Moreover, there exist two states
$\tilde{\omega}_1,\tilde{\omega}_2\in B(\C^d)$ and a POVM element
$\tilde{A}\in B(\C^d)$ that saturate the former inequality.

\begin{proof}

Let $\sum_{i}\mu_i\proj{i}$ be the spectral decomposition of $A$,
where $0\leq \mu_i\leq 1$ and $\{\ket{i}\}_{i=0}^{d-1}$ is an
orthonormal basis of $\C^d$. Then, Eq.~(\ref{formu}) follows from
the next chain of inequalities:
\begin{widetext}
\begin{eqnarray}
&\left(\tr\{\omega_1\cdot\omega_2\}\right)^{1/2}=\left(\sum_{i,j}\bra{i}\omega_1\ket{j}\bra{j}\omega_2\ket{i}\right)^{1/2}\leq\left(\sum_{i,j}\left(\bra{i}\omega_1\ket{i}\bra{j}\omega_1\ket{j}\right)^{1/2}\left(\bra{i}\omega_2\ket{i}\bra{j}\omega_2\ket{j}\right)^{1/2}\right)^{1/2}=\nonumber\\
 &=\sum_i\left(\bra{i}\omega_1\ket{i}\bra{i}\omega_2\ket{i}\right)^{1/2}=\sum_{i}\left(\mu_i\bra{i}\omega_1\ket{i}\mu_i\bra{i}\omega_2\ket{i}\right)^{1/2}+\sum_{i}\left((1-\mu_i)\bra{i}\omega_1\ket{i}(1-\mu_i)\bra{i}\omega_2\ket{i}\right)^{1/2}\leq\nonumber\\
 &\leq\left(\sum_{i}\mu_i\bra{i}\omega_1\ket{i}\right)^{1/2}\left(\sum_{i}\mu_i\bra{i}\omega_2\ket{i}\right)^{1/2}+\left(\sum_{i}(1-\mu_i)\bra{i}\omega_1\ket{i}\right)^{1/2}\left(\sum_{i}(1-\mu_i)\bra{i}\omega_2\ket{i}\right)^{1/2}=\nonumber\\
 &=\sqrt{P_1\cdot P_2}+ \sqrt{(1-P_1)\cdot(1-P_2)}.
\end{eqnarray}
\end{widetext}

To see that the bound is optimal, choose $A=\proj{0}$ and notice
that the relations $|\braket{\psi_{1,2}}{0}|^2=P_{1,2}$ and
$|\braket{\psi_{1}}{\psi_{2}}|^2=f(P_1,P_2)^2$ hold for the states
$\ket{\psi_{1,2}}=\sqrt{P_{1,2}}\ket{0}+\sqrt{1-P_{1,2}}\ket{1}$.

\end{proof}

Taking $A=M$ and $\omega_1=\rho_j\otimes\sigma_l$
($\omega_1=\rho_l\otimes\sigma_j$),
$\omega_2=\rho_k\otimes\sigma_l$
($\omega_2=\rho_l\otimes\sigma_k$) in Proposition \ref{prob}, it
is easy to see that

\begin{eqnarray}
\tr(\rho_j\cdot\rho_k)&\leq& \min_l f(P_{jl},P_{kl})^2\nonumber\\
(\tr(\sigma_j\cdot\sigma_k)&\leq& \min_l f(P_{lj},P_{lk})^2).
\end{eqnarray}

\noindent These relations will play an important role when combined with the next proposition.

\begin{prop}
\label{approxi} Let $\omega=\lambda\proj{u}\otimes\proj{v}\in
B(\C^2\otimes\C^2)$, with $\lambda> 0$, let
$\{\rho_j,\sigma_j\}_{j=1}^3$ be normalized qubit states such that
$\tr\{\rho_j\cdot \rho_k\},\tr\{\sigma_j\cdot\sigma_k\}\leq C$,
and let $\tr\{\omega\rho_j\otimes\sigma_j\}\leq c$. Then,
\be
\lambda\leq \frac{4c}{(1-\sqrt{C})^2}.
\label{ojo}
\ee
\end{prop}

\begin{proof}
Let $\rho_j=\frac{\id+\vec{m}_j\cdot\vec{\sigma}}{2}$,
$\sigma_j=\frac{\id+\vec{n}_j\cdot\vec{\sigma}}{2}$. By hypothesis
we have that \be
(1+\vec{u}\cdot\vec{m}_j)\cdot(1+\vec{v}\cdot\vec{n}_j)\leq
\tilde{c}, \ee

\noindent for $j=1,2,3$, where $\tilde{c}:=4c/\lambda$. It follows
that, for each inequality $j$, at least one of the factors on the
left-hand side is smaller or equal than $\sqrt{\tilde{c}}$. This
implies that there exist two indices $j,k\in\{1,2,3\},j\not=k$,
such that either \be (1+\vec{u}\cdot\vec{m}_{j,k})\leq
\sqrt{\tilde{c}} \label{alter1} \ee

\noindent or
\be
(1+\vec{v}\cdot\vec{n}_{j,k})\leq \sqrt{\tilde{c}}
\label{alter2}
\ee

\noindent holds. Let us assume that Eq. (\ref{alter1}) is true. Then we have that
\be
\|\vec{m}_j+\vec{m}_k\|\geq -\vec{u}\cdot(\vec{m}_j+\vec{m}_k)\geq 2-2\sqrt{\tilde{c}},
\ee

\noindent and hence
\be
\lambda\leq \left(\frac{4\sqrt{c}}{2-\|\vec{m}_j+\vec{m}_k\|}\right)^2.
\label{interm1}
\ee

\noindent On the other hand,
\be
\|\vec{m}_j+\vec{m}_k\|^2\leq 2(1+\vec{m}_j\cdot\vec{m}_k)=4\tr(\rho_j\cdot\rho_k)\leq 4C.
\label{interm2}
\ee

\noindent Combining Eqs.~(\ref{interm1}) and (\ref{interm2}), we
arrive at \be \lambda\leq \frac{4c}{(1-\sqrt{C})^2}. \ee

Had we supposed that Eq.~(\ref{alter2}) were true, we would have
obtained the same relation. Eq.~(\ref{ojo}) then follows from the
fact that at least one of the conditions (\ref{alter1}),
(\ref{alter2}) must hold.
\end{proof}

We are now ready to derive the non-linear inequality. Suppose that $M$ is of the form
\be
M=\sum_{i=1}^K\lambda_i\proj{u_i}\otimes\proj{v_i},
\label{separable}
\ee

\noindent with $\lambda_i>0$. It is straightforward that the
relation \be
\lambda_i\bra{u_i}\rho_j\ket{u_i}\bra{v_i}\sigma_j\ket{v_i}\leq
c_i\equiv
\lambda_i\sum_{k=1}^3\bra{u_i}\rho_k\ket{u_i}\bra{v_i}\sigma_k\ket{v_i}
\ee \noindent holds for $j=1,2,3$. Also, notice that
\be \sum_{i}
c_i=P_{11}+P_{22}+P_{33}. \label{completion} \ee

\noindent Identifying $C=R^2$, with
\be
R\equiv\max_{j\not=k}\left\{\min_lf(P_{jl},P_{kl}),\min_lf(P_{lj},P_{lk})\right\},
\label{defR}
\ee

\noindent in Proposition \ref{approxi}, we thus arrive at the
bound $\lambda_i \leq\frac{4c_i}{(1-R)^2}$. Putting all together,
we have that \be \sum_{i=1}^K\lambda_i\leq \frac{4\sum_{i=1}^Kc_i}{(1-R)^2}=\frac{4(P_{11}+P_{22}+P_{33})}{(1-R)^2}, \ee

\noindent where in the last step we made use of
Eq.~(\ref{completion}).

On the other hand, \be
\sum_i\lambda_i\geq\sum_i\lambda_i\tr(\proj{u}\rho_k)\cdot\tr(\proj{v}\sigma_l)=P_{kl},
\ee \noindent for all $k,l=1,2,3$, and so the inequality
\be
\frac{4(P_{11}+P_{22}+P_{33})}{(1-R)^2}-P_{k,l}\geq 0,
\ee

\noindent with $R$ given by Eq.~(\ref{defR}), must hold for all
$k,l$.

It is immediate to see that the example $P$ in Section
\ref{nonlin} violates it by an amount of $-1/4$.

\end{document}